\documentclass[runningheads]{llncs}
\usepackage[T1]{fontenc}
\usepackage{graphicx}
\usepackage{amssymb}
\usepackage{amsmath}
\usepackage{tikz}
\usepackage{caption}
\usepackage{subcaption}

\begin{document}
\title{Perturbation of Fiedler vector: interest for graph measures and shape analysis}
\titlerunning{Perturbation of Fiedler vector}
%
\author{Julien Lefevre\inst{1} \and Justine Fraize\inst{2,3} \and David Germanaud\inst{2,3}}
\authorrunning{J. Lefevre et al}
%
\institute{ Institut de Neurosciences de la Timone, Aix Marseille Univ, UMR CNRS 7289, Marseille\\ 	\email{julien.lefevre@univ-amu.fr} \and 
CEA Paris-Saclay, Joliot Institute, NeuroSpin, UNIACT, Gif-sur-Yvette, France \and  Université Paris Cité, Inserm, U1141 NeuroDiderot, inDEV, Paris, France
}

%
\maketitle              

\begin{abstract}
	In this paper we investigate some properties of the Fiedler vector, the so-called first non-trivial eigenvector of the Laplacian matrix of a graph. There are important results about the Fiedler vector to identify spectral cuts in graphs but far less is known about its extreme values and points.  We propose a few results and conjectures in this direction. We also bring two concrete contributions, i) by defining a new measure for graphs that can be interpreted in terms of extremality (inverse of centrality), ii) by applying a small perturbation to the Fiedler vector of cerebral shapes such as the corpus callosum to robustify their parameterization.
	\keywords{Graph Laplacian, Fiedler vector, Shape analysis}
\end{abstract}

\section{Introduction}
\label{sec:in}

Let $G=(V,E)$ be an undirected graph  where $V=\{v_i \}_{i=1...n}$ are the vertices and $E=\{(v_i,v_j) \}$ the edges. The adjacency matrix $A$ is defined by $A_{i,j}=1$ if $i\neq j$ and $e_{i,j} \in E$. $A_{i,j}=0$ otherwise. The degree matrix $D$ is a diagonal matrix where $D_{i,i}=\deg(v_i):=\sum_{j=1...n} A_{i,j}$. The (unnormalized) graph Laplacian is the matrix $L:=D-A$. $L$ is a symmetric, semi-definite positive matrix. It has $n$ eigenvalues $0=\lambda_1\leq \lambda_2 \leq ... \leq \lambda_n$. The multiplicity of $0$ equals the number of connected components of $G$. In our case we will consider connected graphs and in that case the associated eigenfunction is constant. In this article we will focus more precisely on the second smallest eigenvalue (\emph{algebraic connectivity}) and associated eigenvector that is called Fiedler vector, denoted by $\Phi$. In the following we will assume that $||\Phi||=1$ and the eigenvalue $\lambda_2$ is simple, so $\Phi$ is uniquely defined up to a sign. We have first a very classical and useful result that is obtained from the Courant min-max principle:

\begin{equation}
	\lambda_2 = \min_{ ||X||_2=1 } {}^{\top}XLX = \sum_{i,j} A_{i,j} \big(\Phi(i) - \Phi(j)\big)^2
\end{equation}

Given that eigenvectors are orthogonal and the eigenvector associated to eigenvalue $0$ is constant we have $\sum_i \Phi(i) =0$ and $\Phi$ has sign changes. Since seminal works by Fiedler \cite{fiedler1973algebraic} there has been a considerable amount of theoretical results on spectral properties of graph Laplacian. Even if not of interest in our case, the subgraph of $G$ induced on the vertices $v$ with $\Phi(v) \geq 0$ is connected. This property allows to decompose a graph in two sub-domains, according to the sign of $\Phi$.

In the past ten years stimulating connections have been made between the Fiedler vector and the first non trivial Laplacian eigenfunction $u$ with Neumann boundary conditions on $\partial D$. In this continuous setting it has been postulated since the 70's that the maximum and minimum of $u$ are located on $\partial D$. The underlying \emph{hot-spots conjecture} turns out to be false but it has raised new interests regarding the extreme points of the Fiedler vector. In \cite{chung2011hot} it has been conjectured that such points, for a closed surface with no holes, maximized the geodesic distance. The conjecture has been proved to be false on a specific class of trees called \emph{Rose graphs} \cite{evans2011fiedler,lefevre2013fiedler}. A few recent works have generalized the previous examples to offer better characterization of extreme points of the Fiedler vector for trees \cite{gernandt2019schur,lederman2019extreme}. To state it very rapidly, the most general and simple result we have comes from a theorem by Fiedler \cite{fiedler1975property}(see also \cite{lederman2019extreme}) stating that the Fiedler vector is monotonic along branches of a tree which implies that the maximal and minimal values are attained in vertices with degree 1 (\emph{pendant vertices}).\\


Given this rapid state of the art, it is possible to present our contributions:
\begin{itemize}
	\item We are interested in understanding more the properties of extreme values of Fiedler vector and for that we will use perturbations of graph Laplacian. A natural question is to know whether extreme points of Fiedler vectors are \emph{stable under perturbations} of the graph.
	\item The perturbation we will consider first consists in adding a pendant vertex to any vertex of the graph. Besides we will vary the weight $x$ on the new edge, \emph{not only for small values} but also by looking at the limit $x \rightarrow +\infty$.
	\item In the previous process we can wonder for which value of $x$ the Fiedler vector of the new graph has an extrema on the new vertex.  This will allow us to define a new measure of graphs that can be interpreted in terms of centrality/periphery.
	\item Last we apply the (small) perturbation procedure to characteristic points of medical shapes and demonstrate that it allows to robustify the description of a longitudinal structure such as the corpus callosum.
\end{itemize}

\section{Perturbation of Fiedler vector}

\subsection{Intuitions}

First we can do a basic representation of our situation of interest involving a graph $G$ and a perturbation consisting in adding a weighted edge between a vertex $v$ and a new vertex $n+1$. \\

\begin{tikzpicture}[style=thick,scale=0.8]
	\node[left] at (1.8,-1.5) {$G$};
	\node[left] at (3.3,-2) {$v$};
	\node[circle,fill=black,scale=0.65] (A) at (3.5,-2) {} ;
	\draw[black,dashed] (A) arc (0:360:1.50);
	
	\node[left] at (6.8,-1.5) {$\tilde{G}$};
	\node[left] at (8.3,-2) {$v$};
	\node[above] at (10.5,-1.8) {$n+1$}; 
	\node[circle,fill=black,scale=0.65] (B) at (8.5,-2) {} ;
	\node[circle,fill=black,scale=0.65] (C) at (10.5,-2) {} ;
	\draw[black,dashed] (B) arc (0:360:1.50);
	\draw[black] (B) -- (C);  
	\node[above] at (9.5,-2.4) {$x$}; 
\end{tikzpicture}

\vspace{0.5cm}

The weight $x$ can be interpreted as the inverse of a distance. Namely considering for $G$ the line graph with $n$ vertices, $L(G)$ can be seen as a finite difference approximation of the second derivative operator on a segment sampled by $n$ regularly spaced points $t_i$, since $f''(t_i) \approx f(t_{i+1}) + f(t_{i-1}) - 2 f(t_i)$. 
By adding a perturbation at one end $n$ we obtain the following matrix:\\

\begin{minipage}{0.5\linewidth}

$$ L(\tilde{G}) = \begin{pmatrix}
	1 & -1 & 0 &  ... & 0 &0\\
	-1 & 2 & -1 & ... & 0 & 0\\
	... & ... & ... & ... & ... & 0\\
	0 & ... & -1 & 2 & -1 & 0\\
	0 & ... & 0 & -1 & 1+x & -x\\
	0 & ... & 0 & 0 & -x & x\\
\end{pmatrix} 
$$

\end{minipage}
\hfill
\begin{minipage}{0.45\linewidth}

\includegraphics[width=\linewidth]{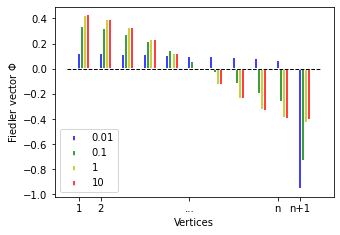}
\captionof{figure}{Evolution of the Fiedler vector when $x$ varies in $[0.01,10]$.}
\label{fig1}
\end{minipage}
\vspace{0.5cm}

We can see that it corresponds approximately to an irregular sampling of the segment $[0,n-1+1/x]$ with $n$ intervals of length $1$ and the last interval of length $1/x$. When $x$ is large we expect the Fiedler vector to be close to a cosine function with only one oscillation on the segment. Besides the Fiedler vector on $n$ and $n+1$ converges to a same value. Conversely, if $x$ is small, the $n$ first points will tend to share the same value of the Fiedler vector and the last point to have the largest magnitude, of opposite sign. It is confirmed by numerical simulations on Fig \ref{fig1}.\\
In the previous case, the graph has been perturbed at a very specific position - one of the two extremities. In the following we will investigate first the more general situation of a perturbation at any vertex.



\subsection{Classical results}

We recall classical results when considering eigenvalues and eigenvectors of a symmetric matrix $M$ and its perturbation by another symmetric matrix $P$ \cite{stewart1990matrix}.

\begin{theorem}[Weyl's inequalities]
	\label{theo1}
	Let $\alpha_1 \geq ... \geq \alpha_n$, $\delta_1 \geq ... \geq \delta_n$, $\gamma_1 \geq ... \geq \gamma_n$ be the spectra of $M$, $P$ and $M+P$ respectively. Then we have:
	\begin{equation}
		\gamma_{i+j-1} \leq \alpha_i + \delta_j \leq \gamma_{i+j-n}
	\end{equation}
\end{theorem}

We have also a local counterpart when the perturbation is small.

\begin{theorem}
	\label{theo2}
	Let $\lambda_1 \leq ... \leq \lambda_n$ and $\tilde{\lambda_1} \leq ... \leq \tilde{\lambda_n} $ be the spectra of matrices $M$ and $\tilde{M}$ ; $\Phi_1,... \Phi_n$ and $\tilde{\Phi_1},... \tilde{\Phi_n}$ corresponding eigenvectors. Then if the eigenvalue $\lambda_i$ is simple, we have the two following approximations:
	\begin{eqnarray}
		\tilde{\lambda_i}&=& \lambda_i+ {}^\top \Phi_i \big(\tilde{M}-M\big) \Phi_i + o\big(||\tilde{M}-M ||\big) \\
		\tilde{\Phi_i}&=&\Phi_i+ \sum_{j \neq i}  \frac{ {}^\top\Phi_j \big(\tilde{M}-M \big) \Phi_i}{\lambda_i-\lambda_j} \Phi_j + o\big(|| \tilde{M}-M  || \big)
	\end{eqnarray}
	
\end{theorem}

In practice, the formula are of little use in our case because the perturbed matrix has the eigenvalue 0 with multiplicity 2 and of course when the perturbation is large.

\subsection{A new result for small perturbations}

\begin{proposition}
	\label{prop1}
	We consider an undirected connected graph $G=(V,E)$ and a Fiedler vector $\Phi$. Given a vertex $v$, we look at a weighted graph $\tilde{G}=(\tilde{V},\tilde{E},\tilde{W})$ where $\tilde{V}=V \cup \{n+1\}$, $\tilde{E}=E \cup \{ (v,n+1)\}$. The weights $\tilde{W}_{i,j}$ are $0$ or $1$ depending on the adjacency between $i$ and $j$ except for  $\tilde{W}_{v,n+1}=x > 0 $. Calling $\Phi(x,\cdot) \in \mathbb{R}^{n+1}$ a Fiedler vector of the graph Laplacian of $\tilde{G}$. Then, there exist $a(v)>0$ that satisfies:
	\begin{equation}
		\label{distance}a(v)=\max\{ a / \forall x \,\ 0\leq x \leq a, \,\  \Phi(x,n+1)= \arg \max_i \Phi(x,i)  \}
\end{equation}
\end{proposition}

\begin{proof}
	First we obtain an upper bound on $\lambda_2(x)$, the algebraic connectivity of $\tilde{G}$. Indeed the graph Laplacian of  $\tilde{G}$ can be expressed as:
	
	$$
	\begin{pmatrix}
		L & \bf{0}\\
		{}^\top \bf{0} & 0\\
	\end{pmatrix} + x \begin{pmatrix}
		 E_{v,v} & -\bf{e}_v\\
		-{}^\top \bf{e}_v & 1\\
	\end{pmatrix}
	$$
	where $\bf{e}_v \in \mathbb{R}^n$ is $0$ everywhere except $1$ on the row $v$ and $E_{v,v} = \bf{e}_v {}^\top \bf{e}_v $.
	Eigenvalues of those two matrices are respectively $0, 0, \lambda_2,... , \lambda_n$ and $0$ (multiplicity $n$), $2\epsilon$. Then by the left inequality in Theorem \ref{theo1}, with $i=n$ and $j=1$ we get $\lambda_2(x) \leq 0 + 2x $.
	Since algebraic connectivity is positive then $\lambda_2(x) \rightarrow 0$ when $ x \rightarrow 0$.\\
	Next we use Courant's theorem:
	$$ \lambda_2(x)=x \big( \Phi(x,n+1)-\Phi(x,v) \big)^2 + \sum_{1\leq i  < j<n+1  }  A_{i,j}\big(\Phi(x,i)-\Phi(x,j)\big)^2$$
	The sum on the right tends to $0$ and since $A_{i,j}\geq 0$ we obtain that as soon as $i$ and $j$ are neighbors and different from $n+1$, $\Phi(x,i)-\Phi(x,j) \rightarrow 0$. 
	But for $i \neq n+1$ $(L \Phi(x, 1:n))(i) = \sum_{j} A_{i,j} \big(\Phi(x,i)-\Phi(x,j)\big)$ and so $(L \Phi(x,1:n))(i) \rightarrow 0$. We conclude that $\Phi(x,1:n)$ converges to the eigenspace of $L$ associated to eigenvalue $0$, i.e. the span of the constant vector.
	
	By using the fact that $||\Phi(x,\cdot)||=1$ and $\sum_{i=1}^{n+1}\Phi(x,i)=0$ we get that:
	$$ \Phi(x,\cdot) \rightarrow \frac{1}{\sqrt{n(n+1)}}\begin{pmatrix} -\bf{1}\\ n \end{pmatrix}$$
	By continuity of $\Phi(x,\cdot)$, it is possible to find an interval $[0,a(v)[$ where $n+1$ remains a maximum point.  \hfill $\qed$ 
	
\end{proof}

Since our previous result is independent from the choice of $v$ one can be naturally interested to know what is the maximum value of $a(v)$.

\subsection{The case of large perturbations}

Here we consider the situation where $x$ is large. It is interesting to examine first the case of complete graph with $n$ vertices. As before we add a vertex $n+1$ to the vertex $n$ with weight $x$ large. By arguments of symmetry it is reasonable to look at a perturbed Fiedler vector of the form ${}^{\top} (-1,...,-1,a,b)$ (up to a constant). The $n-1$ first lines of the eigenvalue problem are the same: $\lambda_2(x) = a+1$. The last lines yields $-x(a-b) = \lambda_2(x) b $. Rearranging all those terms and considering that $n-1-a-b = 0$ we obtain that $a$ should be one of the root of the polynomial $a^2-a(2x+n-2)+(n-1)(x-1)$. Asymptotically one is like $2x$ and the other one tends to $(n-1)/2$. So $\lambda_2(x) \rightarrow (n+1)/2$ with associated Fiedler vector ${}^{\top} (-1,...,-1,(n-1)/2,(n-1)/2)$. Given that $b-a >0$ we conclude that $n+1$ is an extremum of the Fiedler vector for $x$ sufficiently large.\\
The empirical result observed for the line graph is preserved here and we can propose the following conjecture:

\begin{conjecture}
Given an undirected connected graph $G$. We consider $v$ an extremum of the Fiedler vector of the graph $G$. $\tilde{G}$ is the graph obtained from $G$ and $v$ as in Proposition \ref{prop1}. Then for all $x>0$ the Fiedler vector $\Phi(x,\cdot)$ of $\tilde{G}$ has an extremum at $n+1$.
\end{conjecture}

%
%

\section{Applications }

\subsection{A new measure for graphs}
\label{sec3}

\begin{definition}
We consider an undirected connected graph $G=(V,E)$. Given a vertex $v$ we consider $a(v)>0$ as defined in equation \ref{distance} of Proposition \ref{prop1}. We will denote $Fcd(v):=1/a(v)$ the Fiedler centrality distance (Fcd), which is a finite and positive number.	
\end{definition}

\begin{figure}
	\includegraphics[width=\linewidth]{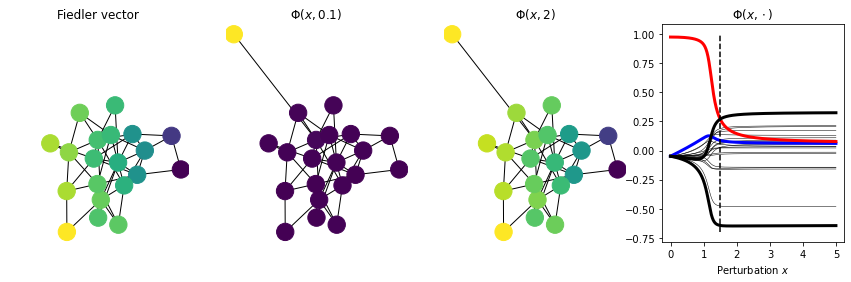}
	\caption{From left to right: Fiedler vector of a Erdös-Renyi random graph, values are encoded from blue ($-$) to yellow ($+$) ; Fiedler vector of the perturbed graphs with $x=0.1$ and $x=2$ ; Plot where each curve corresponds to $\Phi(x,i)$ with $i$ fixed. In blue $i = v$, in red $i=n+1$ and in black the vertices of the two extrema of the initial Fiedler vector. The dotted line corresponds to the value $a(v)$ after which $n+1$ is no more an extremum of $\Phi(x,\cdot)$ }
	\label{fig2}
\end{figure}


Following the previous conjecture we expect that $d(v) = 0$ if $v$ is an extremum of the Fiedler vector of $G$. This measure is supposed to reflect a distance to what could be a boundary. 

On Figure \ref{fig2} we illustrate the evolution of $\Phi(x,\cdot)$ on an Erdös-Renyi random graph $G(n,m)$ with $n=20$ vertices and $m=45$ edges. Note that when $x$ exceeds the threshold $a(v)$, $\Phi(x,1:n)$ is very similar to the Fiedler vector of the unperturbed graph.

\paragraph{Implementation aspects}
We can propose a variation on the previous definition by considering the quantity 
 	$
 	\bar{a}(v)=\max\{x>0 / \,\  \Phi(x,n+1)= \arg \max_i \Phi(x,i)  \}
$
Clearly $a(v) \leq \bar{a}(v)$ and we conjecture that $a(v) = \bar{a}(v)$
based on empirical observations. If this conjecture is true, it allows a fast computation of $a(v)$ based on a dichotomous search in an interval $[x_{min}=10^\alpha,x_{max}=10^\beta]$ then iterating by computing the geometric mean $\bar{x}$ of the extremities and choosing the good side depending if $v$ is an extrema of $\Phi(\bar{x},\cdot)$. It requires at most $K=\log_2(\beta/\alpha)$ steps.

\setcounter{footnote}{0}

Thus, computing $a(v)$ needs $K$ steps where a symmetric eigenproblem is solved. In practice we have used the function \texttt{eigh} of \texttt{scipy} to obtain the two first eigenpairs. Experiments on random graphs $G(n,m)$ with $n$ varying in $[10,200]$
and $m = p n$ with $p$ varying in $[0.2,0.9]$ have revealed that the time complexity is between $O(n^2)$ and $O(n^3)$ which is consistent with existing results \footnote{For instance on \url{https://stackoverflow.com/questions/50358310/how-does-numpy-linalg-eigh-vs-numpy-linalg-svd}}. All the codes used for the article are available on github \footnote{\url{https://github.com/JulienLefevreMars/GSI\_2023}}.

\begin{figure}
	\includegraphics[width=0.3\linewidth]{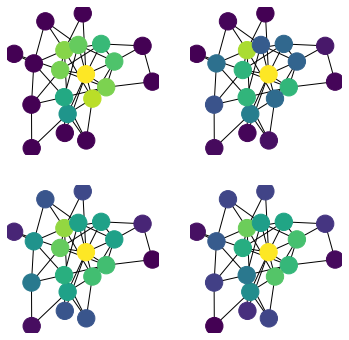}
	\hfill
	\includegraphics[width=0.5\linewidth]{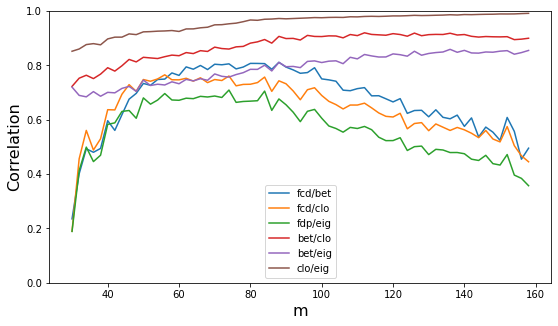}\\
	\caption{Left: Same graph as in Figure \ref{fig2} with Fiedler centrality distance, betweenness centrality, closeness centrality and eigenvector centrality. Right: Correlations between 3 centrality measures and ours (fcd) for the $G(n,m)$ model with $n=20$ and $m$ varying in $[30,160]$ . }
	\label{fig3}
\end{figure}

\paragraph{Comparison with other centrality measures}
Next we can compare our new centrality measure with existing ones such as betweenness centrality, closeness centrality and eigenvector centrality. Figure \ref{fig3} left illustrates visually the similarity between those 4 measures. We also generate $N=100$ random graphs from the $G(n,m)$ model and compute all the correlations between the $4$ measures. For $n=20$ we observe on Figure \ref{fig3} right the evolution of the correlations with $m$ in the same spirit as in \cite{li2015correlation}. There are differences between our measure and the 3 others when the graphs are fully or weakly connected and a good correlation in between. 

Those preliminary results suggest to test the Fiedler distance centrality on real networks to see whether complementary information can be obtained with respect to classical centralities.

Finally, we would like to draw the reader's attention to an important point. Worst case complexity of centrality algorithms is $O(n^3)$ which makes them not scalable on very large graphs. Our method is no exception to this situation and it is tempting to follow the general trend to use deep learning methods to approximate the centrality metrics \cite{wandelt2020complex}. At this stage, an essential question is to know the interest and benefits of this choice, especially with regard to the risks linked to the massification of deep learning and its environmental and societal impacts \cite{ligozat2022unraveling}. 

%

\subsection{Longitudinal parameterization of the corpus callosum}


Finally, we show a very practical and useful application of the previous theoretical framework in the context of shape morphometry.
The corpus callosum is a cerebral structure composed of axons of the two hemispheres joining in the center of the brain. The corpus callosum is easily visualized in MRI brain imaging, on medial sagittal slices. This structure can be affected in some neurological disorders, as in the fetal alcohol syndrome. In their study, the authors \cite{fraize2022combining} measured manually the thickness of the corpus callosum. They needed to replicate their results by making fully automated measurements of this geometrical 2D shape. 

From the MRI acquisitions a segmentation of the corpus callosum is obtained and it is possible to build a planar graph modelling this 2D shape. Given the elongated shape of the corpus callosum (see Figure \ref{fig4}) we use the Fiedler vector of this graph to compute a quasi-isometric parameterization \cite{coulon2015quasi}. Figure \ref{fig4} illustrates how we can obtain at the end a map of the corpus callosum thickness on regularly spaced slices following the longitudinal orientation of the shape given by the level sets of the Fiedler vector.

\begin{figure}
	\includegraphics[width=0.47\linewidth]{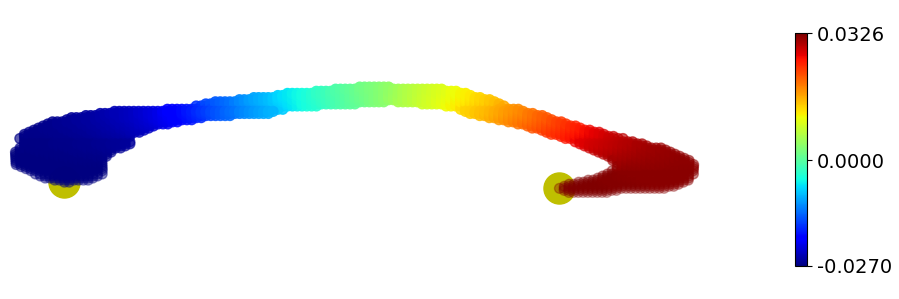}
	\includegraphics[width=0.47\linewidth]{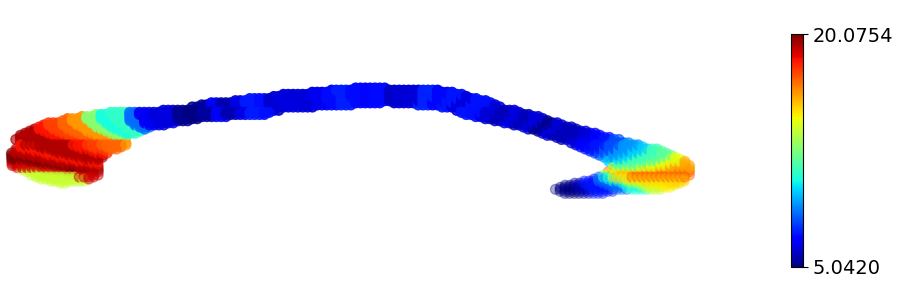}
	\caption{Left: Fiedler vector of the corpus callosum $S_1$ and the two extrema in yellow. Right: thickness map on each isoline.}
	\label{fig4}
\end{figure}

The anatomical main axis of the corpus callosum is well described by the Fiedler vector but among the $125$ processed shapes, $38$ showed a discordance between the maximum of the Fiedler vector and the tip of the corpus callosum as defined by the expert neurologist. This situation is illustrated on Figure \ref{fig5}. 
Then it is possible to inject the information of the correct position $v$ of this maximum by perturbating the graph Laplacian by adding a pendant vertex at $v$ with a weight close to $a(v)$. Then the new Fiedler vector follows the correct elongation of the corpus callosum as shown on Figure \ref{fig5} left. Eventually the parameterization procedure can be applied without any adaptation. On Figure \ref{fig5} right we can observe that the unperturbed thickness of $S_2$ has a value at the tip much more comparable to the one of $S_1$. The thickness remains almost the same on the rest of the shape for $S_2$ with and without perturbations.\\ 
Our approach allows a more realistic evaluation of the thickness which will benefit group studies of corpus callosum shapes in a future work.

\begin{figure}
\begin{subfigure}[b]{0.48\textwidth}
\includegraphics[width=\textwidth]{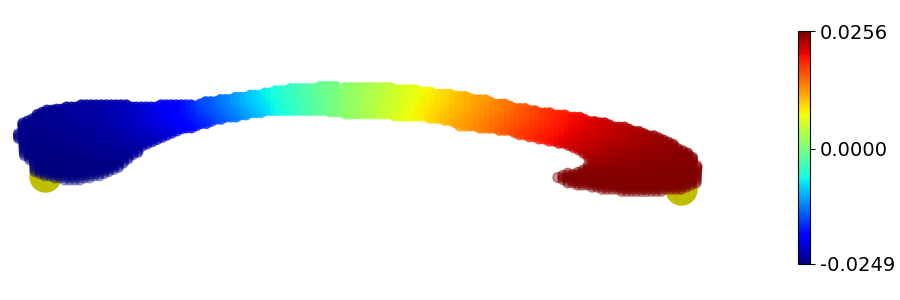}
\includegraphics[width=\textwidth]{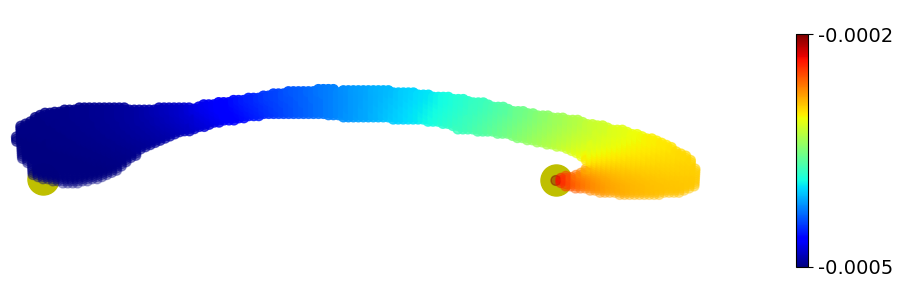}
\end{subfigure}
\begin{subfigure}[b]{0.43\textwidth}
\includegraphics[width=\textwidth]{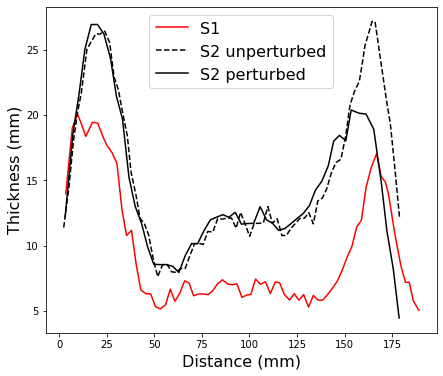}
\end{subfigure}
\caption{Top Left: example of corpus callosum $S_2$ where the maximum of Fiedler vector (yellow dot on the right) is not correctly located. Bottom left: perturbed Fiedler vector from the correct position of the maximum. Right: Thickness profiles for the two corpus callosum $S_1$ and $S_2$. In black the same shape with the perturbed and unperturbed parameterization.}
\label{fig5}
\end{figure}


\subsubsection{Acknowledgements} This project is funded by the French National Agency for Research (ANR-19-CE17-0028-01) and the French National Institute for Public Health research (IRESP-19-ADDICTIONS-08).

%
%
%
%
%
\bibliographystyle{splncs04}
\bibliography{biblio_rose}
\end{document}